\documentclass[a4paper,UKenglish]{lipics-v2016}

\usepackage{microtype}

\usepackage{graphicx}
\usepackage{amsmath}
\usepackage{amsthm}
\usepackage{amssymb}
\usepackage{amssymb}
\usepackage{color}
\definecolor{gray}{rgb}{0.5,0.5,0.5}
\usepackage[ruled,vlined,linesnumbered]{algorithm2e}

\def\idrm#1{\ensuremath{\mathrm{#1}}}

\newcommand{\no}[1]{}

\newcommand{\todo}[1]{} 

\newcommand{\oT}{{\overline T}}
\newcommand{\ocT}{{\overline {\cal T}}}

\newcommand{\ou}{{\overline u}}
\newcommand{\ov}{{\overline v}}

\newcommand{\cT}{{\cal T}}

\newcommand{\oB}{{\overline B}}

\newcommand{\oXu}{\overline{X_u}}
\newcommand{\oP}{{\overline P}}

\newcommand{\ra}{\idrm{rank}}
\newcommand{\sel}{\idrm{select}}
\newcommand{\acc}{\idrm{access}}

\newcommand{\Acc}{\mathit{Acc}}
\newcommand{\SA}{\mathit{SA}}
\newcommand{\longver}[1]{}

\title{Fast Compressed Self-Indexes with Deterministic Linear-Time Construction%
\footnote{Funded with Basal Funds FB0001, Conicyt, Chile.}}

\author[1]{J. Ian Munro}
\author[2]{Gonzalo Navarro}
\author[1]{Yakov Nekrich}

\affil[1]{Cheriton School of Computer Science, University of Waterloo. 
 {\tt imunro@uwaterloo.ca}, {\tt yakov.nekrich@googlemail.com}.}
\affil[2]{CeBiB --- Center of Biotechnology and Bioengineering, Department of 
Computer Science, University of Chile. {\tt gnavarro@dcc.uchile.cl}.}

\Copyright{J. Ian Munro, Gonzalo Navarro, Yakov Nekrich}


    
\begin{document}

\maketitle

\begin{abstract}
We introduce a compressed suffix array representation that, on a text 
$T$ of length $n$ over an alphabet of size $\sigma$, can be built in $O(n)$ 
deterministic time, within $O(n\log\sigma)$ bits of working space, and counts 
the number of occurrences of any pattern $P$ in $T$ in time 
$O(|P| + \log\log_w \sigma)$ on a RAM machine of $w=\Omega(\log n)$-bit words.
This new index outperforms all the other compressed indexes that can
be built in linear deterministic time, and some others. The only
faster indexes can be built in linear time only in expectation, or require
$\Theta(n\log n)$ bits.
We also show that, by using $O(n\log\sigma)$ bits, we can build in 
linear time an index that counts in time $O(|P|/\log_\sigma n + 
\log n(\log\log n)^2)$, which is RAM-optimal for $w=\Theta(\log n)$ and sufficiently long patterns.
\end{abstract}

\section{Introduction}
\label{sec:intro}

The string indexing problem consists in preprocessing a string $T$ so that,
later, we can efficiently find occurrences of patterns $P$ in $T$. The most
popular solutions to this problem are suffix trees \cite{Weiner73} and suffix
arrays \cite{MM93}. Both can be built in $O(n)$ deterministic time on a text
$T$ of length $n$ over an alphabet of size $\sigma$, and the best variants can 
count 
the number of times a string $P$ appears in $T$ in time $O(|P|)$, and even in
time $O(|P|/\log_\sigma n)$ in the word-RAM model if $P$ is given packed into
$|P|/\log_\sigma n$ words \cite{NN17}. Once counted, each occurrence can
be located in $O(1)$ time. Those optimal times, however, come with two important
drawbacks:
\begin{itemize}
\item The variants with this counting time cannot be built in $O(n)$ worst-case time.
\item The data structures use $\Theta(n\log n)$ bits of space.
\end{itemize}

The reason of the first drawback is that some form of perfect hashing is always
used to ensure constant time per pattern symbol (or pack of symbols). The 
classical suffix trees and arrays with linear-time deterministic construction 
offer $O(|P|\log\sigma)$ or $O(|P|+\log n)$ counting time, respectively.
More recently, those times have been reduced to $O(|P|+\log\sigma)$
\cite{CKL15} and even to $O(|P|+\log\log\sigma)$ \cite{FG13}. Simultaneously 
with our work, a suffix tree variant was introduced by Bille et
al.~\cite{BGS17}, which can be built in linear deterministic time and counts in
time $O(|P|/\log_\sigma n + \log |P| + \log\log\sigma)$.
All those indexes, however, still suffer from the second drawback,
that is, they use $\Theta(n\log n)$ bits of space. This makes them impractical
in most applications that handle large text collections.

Research on the second drawback dates back to almost two decades \cite{NM06},
and has led to indexes using $nH_k(T)+o(n(H_k(T)+1))$ bits, where $H_k(T) \le
\log\sigma$ is the $k$-th order entropy of $T$ \cite{Manzini01}, for any
$k \le \alpha\log_\sigma n -1$ and any constant $0<\alpha<1$. That is, the
indexes use asymptotically the same space of the compressed text, and
can reproduce the text and search it; thus they are called self-indexes. The 
fastest compressed self-indexes that can be built in linear deterministic time 
are able to count in time $O(|P|\log\log\sigma)$ \cite{BCGNNalgor13} or
$O(|P|(1+\log_w \sigma))$ \cite{BelazzouguiN15}. There exist
other compressed self-indexes that obtain times $O(|P|)$ \cite{BelazzouguiN14}
or $O(|P|/\log_\sigma n + \log_\sigma^\epsilon n)$ for any constant $\epsilon>0$
\cite{GrossiV05}, but both rely on perfect hashing and are not built in
linear deterministic time.
All those compressed self-indexes use $O(n\frac{\log n}{b})$ further
bits to locate the position of each occurrence found in time $O(b)$, and to 
extract any substring $S$ of $T$ in time $O(|S|+b)$.

In this paper we introduce the first compressed self-index that can be built
in $O(n)$ deterministic time (moreover, using $O(n\log\sigma)$ bits of space
\cite{MNN17}) and with counting time $O(|P|+\log\log_w \sigma)$, where
$w=\Omega(\log n)$ is the size in bits of the computer word. More precisely,
we prove the following result.

\begin{theorem}
  \label{thm:main}
On a RAM machine of $w=\Omega(\log n)$ bits, we can construct an index for a 
text $T$ of length $n$ over an alphabet of size $\sigma=O(n/\log n)$ in $O(n)$ 
deterministic time using $O(n\log\sigma)$ bits of working space. This index 
occupies $nH_k(T)+ o(n\log\sigma)+O(n\frac{\log n}{b})$ bits of space for a 
parameter $b$ and any $k \le \alpha \log_\sigma n-1$, for any constant 
$0<\alpha<1$. The occurrences of a pattern string $P$ can be counted in 
$O(|P|+\log\log_w \sigma)$ time, and then each such occurrence can be
located in $O(b)$ time. An arbitrary substring $S$ of $T$  can be 
extracted in time $O(|S|+b)$.
\end{theorem}

We obtain our results with a combination of the compressed suffix tree $\cT$
of $T$ and the Burrows-Wheeler transform $\oB$ of the reversed text $\oT$.
We manage to simulate the suffix tree traversal for $P$, simultaneously
on $\cT$ and on $\oB$. With a combination of storing deterministic
dictionaries and precomputed rank
values for sampled nodes of $\cT$, and a constant-time method to compute an
extension of partial rank queries that considers small ranges in $\oB$, we
manage to ensure that all the suffix tree steps, except one, require constant
time. The remaining one is solved with general rank queries in time 
$O(\log\log_w \sigma)$. As a byproduct, we show that the compressed sequence
representations that obtain those $\ra$ times \cite{BelazzouguiN15} can also be 
built in linear deterministic time.

Compared with previous work, other indexes may be faster at counting, but
either they are not built in linear deterministic time 
\cite{BelazzouguiN14,GrossiV05,NN17} or they are not compressed
\cite{NN17,BGS17}. Our index outperforms all the previous compressed 
\cite{FMMN07,BCGNNalgor13,BelazzouguiN15}, as well as some uncompressed
\cite{FG13}, indexes that can be built deterministically.

As an application of our tools, we also show that an index using
$O(n\log\sigma)$ bits of space can be built in linear deterministic
time, so that it can count in time $O(|P|/\log_\sigma n+\log n(\log\log n)^2)$,
which is RAM-optimal for $w=\Theta(\log n)$ and sufficiently long patterns.
Current indexes obtaining similar counting time require $O(n\log\sigma)$ 
construction time \cite{GrossiV05} or higher \cite{NN17}, or 
$O(n\log n)$ bits of space \cite{NN17,BGS17}.

\section{Related Work}

\begin{table}
\begin{center}
\begin{tabular}{c|c|c|c}
& Compressed & Compact & Uncompressed \\
\hline
\rotatebox{90}{~Deterministic~} & 
\begin{minipage}{3.1cm}
\vspace*{-1.8cm}
\textcolor{gray}{$|P|\log\log\sigma$ \cite{BCGNNalgor13}} \\
\textcolor{gray}{$|P|(1+\log_w \sigma)$ \cite{BelazzouguiN15}} \\
$|P|+\log\log_w \sigma$ {\bf (ours)}
\end{minipage} 
&
\begin{minipage}{4.0cm}
\vspace*{-1.8cm}
$|P|/\log n+\log^\epsilon n$ \cite{GrossiV05} \\ (constant $\sigma$) \\
$|P|/\log_\sigma n+\log n(\log\log n)^2$ {\bf (ours)} 
\end{minipage} 
& 
\begin{minipage}{4.5cm}
\vspace*{-1.8cm}
\textcolor{gray}{$|P|+\log\log\sigma$ \cite{FG13}} \\
$|P|/\log_\sigma n + \log |P| + \log\log\sigma$ \cite{BGS17}
\end{minipage} \\
\hline
\rotatebox{90}{Randomized~~} & 
\begin{minipage}{3.1cm}
\vspace*{-1.3cm}
\textcolor{gray}{$|P|(1+\log\log_w \sigma)$ \cite{BelazzouguiN15}} \\
$|P|$ \cite{BelazzouguiN14} \\ 
\end{minipage} 
& 
&
\begin{minipage}{4.5cm}
\vspace*{-1.5cm}
$|P|/\log_\sigma n + \log_\sigma^\epsilon n$ \cite{GrossiV05,NN17}  
\end{minipage} \\
\end{tabular}
\end{center}
\caption{Our results in context. The $x$ axis refers to the space
used by the indexes (compressed meaning $nH_k(T)+o(n\log\sigma)$ bits, compact
meaning $O(n\log\sigma)$ bits, 
and uncompressed meaning $\Theta(n\log n)$ bits), and the $y$ axis refers to the
{\em linear-time} construction. In the cells we show the counting time for a 
pattern $P$. We only list the dominant alternatives, graying out those 
outperformed by our new results.}
\label{tab:res}
\end{table}

Let $T$ be a string of length $n$ over an alphabet of size $\sigma$ that is 
indexed
to support searches for patterns $P$. It is generally assumed that $\sigma =
o(n)$, a reasonable convention we will follow. Searches typically require to 
{\em count} the number of times $P$ appears in $T$, and then {\em locate} the 
positions of $T$ where $P$ occurs. The vast majority of the indexes for this 
task are suffix tree \cite{Weiner73} or suffix array \cite{MM93} variants.

The suffix tree can be built in linear deterministic time 
\cite{Weiner73,McCreight76,Ukkonen95}, even on arbitrarily large integer 
alphabets \cite{Farach97}.
The suffix array can be easily derived from the suffix tree in linear time, but
it can also be built independently in linear deterministic time 
\cite{KA05,KSPP05,KarkkainenSB06}. In their basic forms, these structures
allow counting the number of occurrences of a pattern $P$ in 
$T$ in time $O(|P|\log\sigma)$ (suffix tree) or $O(|P|+\log n)$ (suffix array).
Once counted, the occurrences can be located in constant time each.

Cole et al.~\cite{CKL15} introduced the {\em suffix trays}, a 
simple twist on suffix trees that reduces their counting time to
$O(|P|+\log\sigma)$. Fischer and Gawrychowski \cite{FG13} introduced
the {\em wexponential search trees}, which yield suffix trees with 
counting time $O(|P|+\log\log\sigma)$ and support dynamism.

All these structures can be built in linear deterministic time, but require
$\Theta(n\log n)$ bits of space, which challenges their practicality when
handling large text collections.

Faster counting is possible if we resort to perfect hashing and give away
the linear deterministic construction time. In the classical suffix tree, we
can easily achieve $O(|P|)$ time by hashing the children of suffix tree nodes,
and this is optimal in general. In the RAM model with word size 
$\Theta(\log n)$, and if the consecutive symbols of $P$ come packed into
$|P|/\log_\sigma n$ words, the optimal time is instead $O(|P|/\log_\sigma n)$.
This optimal time was recently reached by Navarro and Nekrich \cite{NN17}
(note that their time is not optimal if $w=\omega(\log n)$), with a simple
application of weak-prefix search, already hinted in the original article
\cite{BBPV10}. However, even the randomized construction time of the 
weak-prefix search structure is $O(n\log^\epsilon n)$, for any constant 
$\epsilon>0$. By replacing the weak-prefix search with the solution of Grossi
and Vitter \cite{GrossiV05} for the last nodes of the search, and using 
a randomized construction of their perfect hash functions, the index of 
Navarro and Nekrich \cite{NN17} can be built in linear randomized time and 
count in time $O(|P|/\log_\sigma n + \log^\epsilon_\sigma n)$.
Only recently, simultaneously with our work, a deterministic
linear-time construction algorithm was finally obtained for an index obtaining
$O(|P|/\log_\sigma n + \log|P| + \log\log\sigma)$ counting time \cite{BGS17}.

Still, these structures are not compressed. Compressed suffix trees and arrays
appeared in the year 2000 \cite{NM06}. To date, they take the space of the
compressed text and replace it, in the sense that they can extract any
desired substring of $T$; they are thus called self-indexes. The space 
occupied is measured in terms of the $k$-th
order empirical entropy of $T$, $H_k(T) \le \log\sigma$ \cite{Manzini01}, which
is a lower bound on the space reached by any statistical compressor that encodes
each symbol considering only the $k$ previous ones. Self-indexes may occupy as
little as
$nH_k(T) + o(n(H_k(T)+1))$ bits, for any $k \le \alpha\log_\sigma n -1$, for
any constant $0<\alpha<1$.

The fastest self-indexes with linear-time deterministic construction
are those of Barbay et al.~\cite{BCGNNalgor13}, which counts in time 
$O(|P|\log\log\sigma)$, and Belazzougui and Navarro 
\cite[Thm.~7]{BelazzouguiN15}, which counts in time $O(|P|(1+\log_w \sigma))$. 
The latter requires $O(n(1+\log_w \sigma))$ construction time, but if 
$\log \sigma=O(\log w)$, its counting time is $O(|P|)$ and its construction 
time is $O(n)$.

If we admit randomized linear-time constructions, then Belazzougui and Navarro
\cite[Thm.~10]{BelazzouguiN15} reach $O(|P|(1+\log\log_w\sigma))$ 
counting time. At the expense of $O(n)$ further bits, in another work
\cite{BelazzouguiN14} they reach $O(|P|)$ counting time.
Using $O(n\log\sigma)$ bits, and if $P$ comes in
packed form, Grossi and Vitter \cite{GrossiV05} can count in time 
$O(|P|/\log_\sigma n + \log_\sigma^\epsilon n)$, for any constant $\epsilon>0$,
however their construction requires $O(n\log\sigma)$ time.

Table~\ref{tab:res} puts those results and our contribution in context.
Our new self-index, with $O(|P|+\log\log_w \sigma)$ counting time, linear-time
deterministic construction, and $nH_k(T)+o(n\log\sigma)$ bits of space, 
dominates all the compressed indexes with linear-time deterministic 
construction \cite{BCGNNalgor13,BelazzouguiN15}, as well as some uncompressed 
ones \cite{FG13} (to 
be fair, we do not cover the case $\log\sigma=O(\log w)$, as in this case the 
previous work \cite[Thm.~7]{BelazzouguiN15} already obtains our result). Our 
self-index also dominates a previous one with linear-time randomized 
construction \cite[Thm.~10]{BelazzouguiN15}, which we
incidentally show can also be built deterministically. The only
aspect in which some of those dominated indexes may outperform ours is in that 
they may use $o(n(H_k(T)+1))$ \cite[Thm.~10]{BelazzouguiN15} or $o(n)$ 
\cite[Thm.~7]{BelazzouguiN15} bits
of redundancy, instead of our $o(n\log\sigma)$ bits.
We also derive a compact index (i.e., using $O(n\log\sigma)$ bits) that is
built in linear deterministic time and counts in time 
$O(|P|/\log_\sigma n + \log n(\log\log n)^2)$, which is the only one in this
category unless we consider constant $\sigma$ for Grossi and Vitter
\cite{GrossiV05}.

\section{Preliminaries}
\label{sec:prelim}

We denote by $T[i..]$ the suffix of $T[0,n-1]$ starting at position $i$ and by $T[i..j]$ the substring that begins with $T[i]$ and ends with $T[j]$, 
$T[i..]=T[i]T[i+1]\ldots T[n-1]$ and $T[i..j]=T[i]T[i+1]\ldots T[j-1]T[j]$. We
assume that the text $T$ ends with a special symbol \$ that lexicographically
precedes all other symbols in $T$.  The alphabet size is $\sigma$ and symbols
are integers in $[0..\sigma-1]$ (so \$ corresponds to $0$).  In this paper, as
in the previous work on this topic, we use the word RAM model of computation.
A machine word consists of $w=\Omega(\log n)$ bits and we can execute standard
bit and arithmetic operations in constant time. We assume for simplicity that 
the alphabet size $\sigma = O(n/\log n)$ (otherwise the text is almost 
incompressible anyway \cite{Gag06}). We also assume $\log\sigma=\omega(\log w)$,
since otherwise our goal is already reached in previous work 
\cite[Thm.~7]{BelazzouguiN15}.

\subsection{Rank and Select Queries}
\label{sec:ranksel}

We define three basic queries on sequences.
Let $B[0..n-1]$ be a sequence of symbols over alphabet $[0..\sigma-1]$. 
The rank query, $\ra_a(i,B)$, counts how many times $a$ occurs among the first 
$i+1$ symbols in $B$, $\ra_a(i,B)=|\{\,j \le i,~ B[j]=a\}|$.  The
select query, $\sel_a(i,B)$, finds the position in $B$ where $a$ occurs for the
$i$-th time, $\sel_a(i,B)=j$ iff $B[j]=a$ and $\ra_a(j,B)=i$. The third 
query is $\acc(i,B)$, which returns simply $B[i]$. 

We can answer $\acc$ queries in $O(1)$ time and $\sel$ queries in any 
$\omega(1)$ time, or vice versa, and $\ra$ queries in time 
$O(\log\log_w \sigma)$, which is optimal 
\cite{BelazzouguiN15}. These structures use $n\log\sigma + o(n\log\sigma)$ bits,
and we will use variants that require only compressed space. In this paper, we 
will show that those structures can be built in linear deterministic time.

An important special case of $\ra$ queries is the partial rank query,
$\ra_{B[i]}(i,B)$, which asks how many times $B[i]$
occurrs in $B[0..i]$. Unlike general rank queries, partial rank queries can
be answered in $O(1)$ time~\cite{BelazzouguiN15}. Such a structure can be
built in $O(n)$ deterministic time and requires $O(n\log\log\sigma)$ bits
of working and final space \cite[Thm.~A.4.1]{MNN17}. 

For this paper, we define a generalization of partial rank queries called 
interval rank queries, $\ra_a(i,j,B) =
\langle \ra_a(i-1,B),\ra_a(j,B)\rangle$, from where in particular we can
deduce the number of times $a$ occurs in $B[i..j]$. If $a$ does not occur
in $B[i..j]$, however, this query just returns $null$ (this is why it can be
regarded as a generalized partial rank query).

In the special case where the alphabet size is small, $\log \sigma=O(\log w)$, 
we can represent $B$ so that $\ra$, $\sel$, and $\acc$ queries are answered in 
$O(1)$ time \cite[Thm.~7]{BelazzouguiN15}, but we are not focusing on this 
case in this paper, as the problem has already been solved for this case.

\subsection{Suffix Array and Suffix Tree}

The suffix tree \cite{Weiner73} for a string $T[0..n-1]$ is a compacted
digital tree on the suffixes of $T$, where the leaves point to the starting
positions of the suffixes. We call $X_u$ the string leading to suffix tree 
node $u$. The suffix array \cite{MM93} is an array 
$\SA[0..n-1]$ such that $\SA[i]=j$ if and only if
$T[j..]$ is the $(i+1)$-th lexicographically smallest suffix of $T$. All the
occurrences of a substring $P$ in $T$ correspond to suffixes of $T$ that start
with $P$. These suffixes descend from a single suffix tree node, called the
{\em locus} of $P$, and also occupy a contiguous interval in the suffix array 
$\SA$. Note that the locus of $P$ is the node $u$ closest to the root for
which $P$ is a prefix of $X_u$. If $P$ has no locus node, then it does not
occur in $T$.

\subsection{Compressed Suffix Array and Tree}

A compressed suffix array (CSA) is a compact data structure that provides the
same functionality as the suffix array. The main component of a CSA is the
one that allows determining, given a pattern $P$, the suffix array range
$\SA[i..j]$ of the prefixes starting with $P$. Counting is then solved as
$j-i+1$. For locating any cell $\SA[k]$, and for extracting any substring $S$
from $T$, most CSAs make use of a sampled array $SAM_b$, which contains the
values of $\SA[i]$ such that $\SA[i]\!\!\mod b =0$ or $\SA[i]=n-1$. Here $b$ is
a tradeoff parameter: CSAs require $O(n\frac{\log n}{b})$ further bits 
and can locate in time proportional to $b$ and extract $S$ in time
proportional to $b+|S|$.
We refer to a survey \cite{NM06} for a more detailed description.

A compressed suffix tree \cite{Sadakane07} is formed by a compressed suffix
array and other components that add up to $O(n)$ bits. These include in
particular a representation of the tree topology that supports constant-time
computation of the preorder of a node, its number of children, its $j$-th 
child, its number of descendant leaves, and lowest common ancestors,
among others \cite{NS14}. Computing node preorders 
is useful to associate satellite information to the nodes.

Both the compressed suffix array and tree can be built in $O(n)$ deterministic
time using $O(n\log\sigma)$ bits of space \cite{MNN17}.

\subsection{Burrows-Wheeler Transform and FM-index} \label{sec:bwt}

The Burrows-Wheeler Transform (BWT) \cite{BW94} of a string $T[0..n-1]$ is 
another
string $B[0..n-1]$ obtained by sorting all possible rotations of $T$ and 
writing the last symbol of every rotation (in sorted order). The BWT is
related to the suffix array by the identity
$B[i]=T[(\SA[i]-1)\!\!\mod n]$. Hence, we can build the BWT by sorting the
suffixes and writing the symbols that precede the suffixes in lexicographical
order. 

The FM-index \cite{FerraginaM05,FMMN07} is a CSA that builds on the BWT. 
It consists of the following three main components:
\begin{itemize}
\item 
The BWT $B$ of $T$.
\item
The array $\Acc[0..\sigma-1]$ where $\Acc[i]$ holds the total number of symbols $a < i$ in $T$ (or equivalently, the total number of symbols $a<i$ in $B$).
\item
The sampled array $SAM_b$.
\end{itemize}

The interval of a pattern string $P[0..m-1]$ in the suffix array $\SA$ can be 
computed on the BWT $B$. The interval is computed backwards: 
for $i=m-1,m-2,\ldots$, we identify the interval of $P[i..m-1]$ in $B$.
The interval is initially the whole $B[0..n-1]$.
Suppose that we know the interval $B[i_1..j_1]$ that corresponds to
$P[i+1..m-1]$. Then the interval $B[i_2..j_2]$ that corresponds to $P[i..m-1]$
is computed as $i_2=\Acc[a]+\ra_c(i_1-1,B)$ and $j_2=\Acc[a]+\ra_c(j_1,B)-1$,
where $a=P[i]$.  Thus the interval of $P$ is found by answering $2m$ $\ra$
queries. Any sequence representation offering rank and access queries can then
be applied on $B$ to obtain an FM-index.

An important procedure on the FM-index is the computation of the function $LF$, 
defined as follows: if $\SA[j]=i+1$, then $\SA[LF(j)]=i$. $LF$ can be computed 
with access and partial $\ra$ queries on $B$, 
$LF(j)=\ra_{B[j]}(i,B)+\Acc[B[j]]-1$, and thus constant-time computation of $LF$
is possible. Using $SAM_b$ and $O(b)$ applications of $LF$, we can locate any 
cell $\SA[r]$. A similar procedure allows extracting any substring $S$ of $T$ 
with $O(b+|S|)$ applications of $LF$.


\section{Small Interval Rank Queries} \label{sec:smallrank}

We start by showing how a compressed data structure that supports select
queries can be extended to support a new kind of queries that we dub
\emph{small interval rank queries}. An interval query $\ra_a(i,j,B)$ is a
small interval rank query if $j-i\le \log^2\sigma$.
Our compressed index relies on the following result.

\begin{lemma}
  \label{lemma:interrank}
Suppose that we are given a data structure that supports $\acc$ queries on a
sequence $C[0..m-1]$, on alphabet $[0..\sigma-1]$, in time $t$. Then, using $O(m\log \log \sigma)$ additional bits, we can  support small  interval rank queries on $C$ in $O(t)$ time.
\end{lemma}

\begin{proof}
 We split $C$ into groups $G_i$ of $\log^2\sigma$
consecutive symbols, $G_i=C[i\log^2\sigma .. (i+1)\log^2\sigma-1]$. Let
$A_i$ denote the sequence of the distinct symbols that occur in $G_i$.
Storing $A_i$ directly
would need  $\log \sigma$ bits per symbol. Instead, we encode each element
of $A_i$ as its first position in $G_i$, which needs only $O(\log\log\sigma)$ 
bits. With this encoded sequence, since we have $O(t)$-time access to $C$, 
we have access to any element of $A_i$ in time $O(t)$. In addition, we store 
a succinct SB-tree~\cite{GrossiORR09} on the elements of $A_i$. This structure
uses $O(p\log\log u)$ bits to index $p$ elements in $[1..u]$, and supports
predecessor (and membership) queries in time $O(\log p / \log\log u)$ plus
one access to $A_i$. Since $u=\sigma$ and $p \le \log^2 \sigma$, the query
time is $O(t)$ and the space usage is bounded by $O(m\log\log\sigma)$ bits.

For each $a\in A_i$ we also keep the increasing list $I_{a,i}$ of all the
positions where $a$ occurs in $G_i$. Positions are stored as differences with
the left border of $G_i$: if $C[j]=a$, we store the  difference
$j-i\log^2\sigma$. Hence elements of $I_{a,i}$ can also be stored in 
$O(\log\log\sigma)$ bits per symbol, adding up to $O(m\log\log\sigma)$ bits. 
We also build an SB-tree on top of each $I_{a,i}$ to provide for predecessor
searches.

Using the SB-trees on $A_i$ and $I_{a,i}$, we can answer small interval 
rank queries $\ra_a(x,y,C)$. 
Consider a group  $G_i=C[i\log^2\sigma..(i+1)\log^2\sigma-1]$, 
an index $k$ such that  $i\log^2\sigma \le k \le (i+1)\log^2\sigma$, and a 
symbol $a$. We can find the largest $i\log^2\sigma \le r\le k$ 
such that $C[r]=a$, or determine it does not exist: First we look for the 
symbol $a$ in $A_i$; if $a\in A_i$, we find the predecessor of
$k-i\log^2\sigma$ in $I_{a,i}$. 

Now consider an interval $C[x..y]$ of size at most $\log^2 \sigma$. It intersects 
at most two groups,
$G_i$ and $G_{i-1}$.  We find the rightmost occurrence of symbol $a$ in
$C[x..y]$ as follows. First we look for the rightmost occurrence $y'\le y$ of
$a$ in $G_i$; if $a$ does not occur in $C[i\log^2\sigma.. y]$, we look for the
rightmost occurrence $y'\le i\log^2\sigma-1$ of $a$ in $G_{i-1}$. If this is
$\ge x$, we find
the leftmost occurrence $x'$ of $a$ in $C[x..y]$ using a symmetric procedure.
When $x' \le y'$ are found, we can compute $\ra_a(x',C)$ and $\ra_a(y',C)$
in $O(1)$ time by answering partial rank queries (Section~\ref{sec:ranksel}). 
These are supported in $O(1)$ time and $O(m\log\log\sigma)$ bits.
The answer is then $\langle \ra_a(x',C)-1,\ra_a(y',C)\rangle$, or $null$ if
$a$ does not occur in $C[x..y]$. 
\end{proof}

The construction of the small interval rank data structure is dominated by the
time needed to build the succinct SB-trees \cite{GrossiORR09}. These are simply
B-trees with arity $O(\sqrt{\log u})$ and height $O(\log p / \log\log u)$, where
in each node a Patricia tree for $O(\log\log u)$-bit chunks of the keys are
stored. To build the structure in $O(\log p / \log\log u)$ time per key, we
only need to build those Patricia trees in linear time. Given that the total
number of bits of all the keys to insert in a Patricia tree is 
$O(\sqrt{\log u} \log\log u)$, we do not even need to build the Patricia
tree. Instead, a universal precomputed table may answer any
Patricia tree search for any possible set of keys and any possible pattern,
in constant time. The size of the table is $O(2^{O(\sqrt{\log u} \log\log u)}
\sqrt{\log u}) = o(u)$ bits (the authors \cite{GrossiORR09} actually use a
similar table to answer queries). For our values of $p$ and $u$, the 
construction requires $O(mt)$ time and the universal table is of $o(\sigma)$ 
bits.

\section{Compressed Index}

We classify the nodes of the suffix tree $\cT$ of $T$ into heavy, light, and special,
as in previous work \cite{NN17,MNN17}. Let
$d=\log \sigma$.  A node $u$ of $\cT$ is \emph{heavy} if it has at least $d$
leaf descendants and \emph{light} otherwise.  We say that a heavy node $u$ is
\emph{special} if it has at least two heavy children.

For every special node $u$, we construct a deterministic dictionary 
\cite{HagerupMP01} $D_u$ that contains the labels of all the heavy children of 
$u$: If the $j$th child of $u$, $u_j$, is heavy and the first symbol on the edge
from to $u$ to $u_j$ is $a_j$, then we store the key $a_j$ in $D_u$ with $j$
as satellite data. If a heavy node $u$ has only one heavy child $u_j$ and $d$
or more light children, then we also store the data structure $D_u$
(containing only that heavy child of $u$). If, instead, a heavy node has one heavy child and less than $d$ light children, we just keep the index of the heavy child using $O(\log d)=O(\log\log \sigma)$ bits. 

The second component of our index is the Burrows-Wheeler Transform $\oB$ of the reverse text $\oT$. We store a data structure that supports rank, partial rank, select, and access queries on $\oB$. It is sufficient for us to support $\acc$ and partial rank queries in $O(1)$ time and $\ra$ queries in 
$O(\log\log_w \sigma)$ time.
We also construct the data structure described in Lemma~\ref{lemma:interrank}, 
which supports small interval rank queries in $O(1)$ time. 
Finally, we explicitly store the answers to some rank queries.  Let $\oB[l_u..r_u]$ denote the range of $\oXu$, where $\oXu$ is the reverse of $X_u$, for a suffix tree node $u$. For all data structures $D_u$ and for every symbol $a\in D_u$  we store the values of $\ra_a(l_u-1,\oB)$ and $\ra_a(r_u,\oB)$. 

Let us show how to store the selected precomputed answers to $\ra$ queries in
$O(\log\sigma)$ bits per query. Following a known scheme \cite{GolynskiMR06},
we divide the sequence $\oB$ into chunks of size $\sigma$. For each symbol
$a$, we encode the number $d_k$ of times $a$ occurs in each chunk $k$ in a
binary sequence $A_a=01^{d_0}01^{d_1}01^{d_2}\ldots$. If a symbol
$\oB[i]$ belongs to chunk $k = \lfloor i/\sigma \rfloor$, then
$\ra_a(i,\oB)$ is $\sel_0(k+1,A_a)-k$ plus the number of times $a$ occurs in
$\oB[k\sigma..i]$. The former value is computed in $O(1)$ time with a
structure that uses $|A_a|+o(|A_a|)$ bits \cite{Cla96,Mun96}, whereas the
latter value is in $[0,\sigma]$ and thus can be stored in $D_u$ using just
$O(\log\sigma)$ bits. The total size of all the sequences $A_a$ is $O(n)$ bits.

Therefore, $D_u$ needs $O(\log\sigma)$ bits per element. The total number of elements in all the structures $D_u$ is  equal to the  number of special nodes plus the number of  heavy nodes with  one heavy child and at least $d$ light children. Hence all $D_u$ contain $O(n/d)$ symbols and use 
$O((n/d)\log\sigma)=O(n)$ bits of space. Indexes of heavy children for nodes with only one heavy child and less than $d$ light children add up to $O(n\log\log\sigma)$ bits. 
The structures for partial rank and small interval rank queries on $\oB$ use
$O(n\log\log\sigma)$ further bits. Since we assume that $\sigma$ is $\omega(1)$,
we can simplify $O(n\log\log\sigma)=o(n\log\sigma)$.

The sequence representation that supports $\acc$ and $\ra$ queries on $\oB$ 
can be made to use $nH_k(T)+o(n(H_k(T)+1))$ bits,
by exploiting the fact that it is built on a BWT
\cite[Thm.~10]{BelazzouguiN15}.\footnote{In fact it is $nH_k(\oT)$, but this 
is $nH_k(T)+O(\log n)$ \cite[Thm.~A.3]{FerraginaM05}.} 
We note that they use constant-time $\sel$
queries on $\oB$ instead of constant-time access, so they can use $\sel$ queries
to perform $LF^{-1}$-steps in constant time. Instead, with our partial rank
queries, we can perform $LF$-steps in constant time (recall 
Section~\ref{sec:bwt}), and thus have constant-time $\acc$ instead of 
constant-time $\sel$ on $\oB$ (we actually do not use query $\sel$ at all). 
They avoid this solution because partial rank queries require
$o(n\log\sigma)$ bits, which can be more than $o(n(H_k(T)+1))$, but we are
already paying this price. 

Apart from this space, array $\Acc$ needs $O(\sigma\log n)= O(n)$ bits and
$SAM_b$ uses $O(n\frac{\log n}{b})$. The total space usage of our self-index
then adds up to $nH_k(T)+o(n\log\sigma)+O(n\frac{\log n}{b})$ bits.

\section{Pattern Search}

Given a query  string $P$, we will find in time $O(|P|+\log\log_w \sigma)$  the
range of the reversed string  $\oP$ in $\oB$. A backward search for $P$ in $B$
will be replaced by an analogous backward search for $\oP$ in $\oB$, that is,
we will find the range of $\overline{P[0..i]}$ if the range of
$\overline{P[0..i-1]}$ is known. Let $[l_i..r_i]$ be the range of
$\overline{P[0..i]}$. We can compute $l_i$ and $r_i$ from
$l_{i-1}$ and $r_{i-1}$ as $l_i=\Acc[a]+\ra_a(l_{i-1}-1,\oB)$ and
$r_i=\Acc[a]+\ra_a(r_{i-1},\oB)-1$, for $a=P[i]$. Using our auxiliary data structures on $\oB$ and the additional information stored in the nodes of the suffix tree $\cT$, we can answer the necessary $\ra$ queries in constant time (with one exception).  The idea is to traverse the suffix tree $\cT$ in synchronization with the forward search on $\oB$, until the locus of $P$ is found or we determine that $P$ does not occur in $T$.

Our procedure starts at the root node of $\cT$, with $l_{-1}=0$, $r_{-1}=n-1$, and $i=0$. 
We compute the ranges $\oB[l_i..r_i]$ that correspond to $\overline{P[0..i]}$
for $i=0,\ldots, |P|-1$. Simultaneously, we move down in  the suffix tree. Let $u$ denote the last visited node of $\cT$ and let $a=P[i]$.  We denote by $u_a$ the next node that we must visit in the suffix tree, i.e., $u_a$ is the locus of $P[0..i]$. 
We can compute $l_i$ and $r_i$ in $O(1)$ time if $\ra_a(r_{i-1},\oB)$ and
$\ra_a(l_{i-1}-1,\oB)$ are known.  We will show below that these queries can
be answered in constant time because either (a)  the answers to $\ra$ queries
are explicitly stored in $D_u$ or (b) the $\ra$ query that must be answered is
a small interval $\ra$ query. The only exception is the situation when we move
from a heavy node to a light node in the suffix tree; in this case  the $\ra$
query takes $O(\log\log_w \sigma)$ time. We note that, once we are in a light
node, we need not descend in $\cT$ anymore; it is sufficient to maintain the
interval in $\oB$.

For ease of description we distinguish between the following cases. 
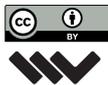
\begin{enumerate}
\item Node $u$ is heavy and $a\in D_u$. In this case we identify the heavy
child $u_a$ of $u$ that is labeled with $a$ in constant time using the
deterministic dictionary. We can also find $l_i$ and $r_i$ in time $O(1)$
because $\ra_a(l_{i-1}-1,\oB)$ and $\ra_a(r_{i-1},\oB)$ are stored  in $D_u$.
\item Node $u$ is heavy and $a\not\in D_u$. In this case $u_a$, if it exists,
is a light node. We then find it with two standard rank queries on $\oB$, in order to compute $l_i$ and $r_i$ or determine that $P$ does not occur in $T$.
\item Node $u$ is heavy but we do not keep a dictionary $D_u$ for the node $u$. In this case $u$ has at most one heavy child and less than $d$ light children. We have two subcases:
	\begin{enumerate}
	\item If $u_a$ is the (only) heavy node, we find this out with a single
comparison, as the heavy node is identified in $u$. However, the values
$\ra_a(l_{i-1}-1,\oB)$ and $\ra_a(r_{i-1},\oB)$ are not stored in $u$. To
compute them, we exploit the fact that the number of non-$a$'s in 
$\oB[l_{i-1}..r_{i-1}]$ is less than $d^2$, as all the children apart from 
$u_a$ are light and less than $d$. Therefore, the first and the last occurrences
of $a$ in $\oB[l_{i-1}..r_{i-1}]$ must be at distance less than $d^2$ from
the extremes $l_{i-1}$ and $r_{i-1}$, respectively. Therefore, a small interval
rank query, $\ra_a(l_{i-1},l_{i-1}+d^2,\oB)$, gives us $\ra_a(l_{i-1}-1,\oB)$, 
since there is for sure an $a$ in the range. Analogously, 
$\ra_a(r_{i-1}-d^2,r_{i-1},\oB)$ gives us $\ra_a(r_{i-1},\oB)$.
	\item If $u_a$ is a light node, we compute $l_i$ and $r_i$ with two 
standard rank queries on $\oB$ (or we might determine that $P$ does not appear
in $T$).
	\end{enumerate}
\item Node $u$ is light. In this case, $P[0..i-1]$ occurs at most $d$ times in
$T$. Hence $\overline{P[0..i-1]}$ also occurs at most $d$ times in $\oT$ and
$r_{i-1}-l_{i-1}\le d$. Therefore  we can compute $r_i$ and $l_i$ in $O(1)$
time by answering a small interval rank query, 
$\langle \ra_a(l_{i-1}-1,\oB),\ra_a(r_{i-1},\oB)\rangle$. If this returns
$null$, then $P$ does not occur in $T$.
\item We are on an edge of the suffix tree between a node $u$ and some child $u_j$ of $u$. 
In this case all the occurrences of $P[0..i-1]$ in $T$ are followed by the same symbol, $c$, and all the occurrences of $\overline{P[0..i-1]}$ are preceded by $c$ in $\oT$. Therefore $\oB[l_{i-1}..r_{i-1}]$ contains only the symbol $c$. 
This situation can be verified with access and partial rank queries
on $\oB$: $\oB[r_{i-1}]=\oB[l_{i-1}]=c$ and
$\ra_c(r_{i-1},\oB) - \ra_c(l_{i-1},\oB) = r_{i-1}- l_{i-1}$.
In this case, if $a \not= c$, then $P$ does not occur in $T$; otherwise we
obtain the new range with the partial rank query $\ra_c(r_{i-1},\oB)$, and
$\ra_c(l_{i-1}-1,\oB) = \ra_c(r_{i-1},\oB) - (r_{i-1}-l_{i-1}+1)$. Note that
if $u$ is light we do not need to consider this case; we may directly apply
case 4.
\end{enumerate}

Except for the cases 2 and 3b, we can find $l_i$ and $r_i$ in $O(1)$ time. In cases 2 and 3b we need $O(\log\log_w \sigma)$ time to answer general $\ra$ queries. However, these cases only take place when the node $u$ is heavy and  its child $u_a$ is light. Since all descendants of a light node are light, those cases occur only once along the traversal of $P$.  Hence the total time to find the range of $\oP$ in $\oB$ is $O(|P|+\log\log_w \sigma)$. Once the range is known, we can count and report all occurrences of $\oP$ in the standard way.

\section{Linear-Time Construction}

\subsection{Sequences and Related Structures}

Apart from constructing the BWT $\oB$ of $\oT$, which is a component of the
final structure, the linear-time construction of the other components requires
that we also build, as intermediate structures, the BWT $B$ of $T$, and the
compressed suffix trees $\ocT$ and $\cT$ of $\oT$ and $T$,
respectively. All these are built in $O(n)$ deterministic time and
using $O(n\log\sigma)$ bits of space \cite{MNN17}.
We also keep, on top of both $\oB$ and $B$, $O(n\log\log\sigma)$-bit data 
structures able to report, for any interval $\oB[i..j]$ or $B[i..j]$, all the 
distinct symbols from this interval, and their frequencies in the interval. 
The symbols are retrieved in arbitrary order. These auxiliary data structures
can also be constructed in $O(n)$ time \cite[Sec.\ A.5]{MNN17}.

On top of the sequences $B$ and $\oB$, we build the representation that
supports $\acc$ in $O(1)$ and $\ra$ in $O(\log\log_w\sigma)$ time
\cite{BelazzouguiN15}. In their original paper, those structures are built 
using perfect hashing, but a deterministic construction is also possible
\cite[Lem.\ 11]{BelazzouguiCKM16}; we give the details next.

The key part of the construction is that, within a chunk of $\sigma$ symbols,
we must build a virtual list $I_a$ of the positions where each symbol $a$ 
occurs, and provide predecessor search on those lists in $O(\log\log_w\sigma)$
time. We divide each list into blocks of $\log^2 \sigma$ elements, and create
a succinct SB-tree \cite{GrossiORR09} on the block elements, much as in
Section~\ref{sec:smallrank}. The search time inside a block is then $O(t)$,
where $t$ is the time to access an element in $I_a$, and the total extra space
is $O(n\log\log\sigma)$ bits. If there is more than one block in $I_a$, then
the block minima are inserted into a predecessor structure 
\cite[App.~A]{BelazzouguiN15} that will find the closest preceding 
block minimum in time $O(\log\log_w \sigma)$ and use $O(n\log\log \sigma)$ 
bits. This structure uses perfect hash functions called $I(P)$, which provide
constant-time membership queries. Instead, we replace them with deterministic 
dictionaries \cite{HagerupMP01}. The only disadvantage of these dictionaries 
is that they require $O(\log\sigma)$ construction time per element, and since 
each element is inserted into $O(\log\log_w \sigma)$ structures $I(P)$, the 
total construction time per element is $O(\log\sigma \log\log_w \sigma)$. 
However, since we build these structures only on $O(n/\log^2 \sigma)$ block
minima, the total construction time is only $O(n)$.

On the variant of the structure that provides constant-time $\acc$, the access
to an element in $I_a$ is provided via a permutation structure \cite{MRRR12}
which offers access time $t$ with extra space $O((n/t)\log\sigma)$ bits. 
Therefore, for any $\log\sigma = \omega(\log w)$, we can have 
$t = O(\log\log_w \sigma)$ with $o(n\log\sigma)$ bits of space.

\subsection{Structures $D_u$}

The most complex part of the construction is to fill the data of the $D_u$
structures. We visit all the nodes of $\cT$ and identify those nodes
$u$ for which the data structure $D_u$ must be constructed. This can be easily 
done in linear time, by using the constant-time computation of the number of 
descendant leaves. To determine if we must build $D_u$, we traverse its 
children $u_1, u_2,\ldots$ and count their descendant leaves to decide if they 
are heavy or light.

We use a bit vector $D$ to mark the preorders of the
nodes $u$ for which $D_u$ will be constructed: If $p$ is the preorder of node
$u$, then it stores a structure $D_u$ iff $D[p]=1$, in which case $D_u$ is 
stored in an array at position $\ra_1(D,p)$. If, instead, $u$ does not store
$D_u$ but it has one heavy child, we store its child rank in another array 
indexed by $\ra_0(D,p)$, using $\log\log\sigma$ bits per cell.


The main difficulty is how to compute the symbols $a$ to be stored in $D_u$,
and the ranges $\oB[l_u,r_u]$, for all the selected nodes $u$. It is not easy 
to do this through a preorder traversal of $\cT$ because we would 
need to traverse edges that represent many symbols. Our approach, instead, is 
inspired by the navigation of the suffix-link tree using two BWTs given by
Belazzougui et al.~\cite{BelazzouguiCKM13}. Let $\cT_w$ denote the tree whose
edges correspond to Weiner links between internal nodes in $\cT$. That is,
the root of $\cT_w$ is the same root of $\cT$ and,
if we have internal nodes $u,v \in \cT$ where $X_v = a \cdot X_u$ for some
symbol $a$, then $v$ descends from $u$ by the symbol $a$ in $\cT_w$.
We first show that the nodes of $\cT_w$ are the internal nodes of $\cT$.
The inclusion is clear by definition in one direction; the other is well-known
but we prove it for completeness.

\begin{lemma}
  \label{lemma:wtree}
All internal nodes of the suffix tree $\cT$ are nodes of $\cT_w$.  
\end{lemma}
\begin{proof}
We proceed by induction on $|X_u|$, where the base case holds by definition.
Now let a non-root internal node $u$ of $\cT$ be labeled by string $X_u = aX$. This means that  there are at least two different symbols $a_1$ and $a_2$ such that both $aXa_1$ and $aXa_2$ occur in the text $T$. Then both $Xa_1$ and $Xa_2$ also occur in $T$. Hence there is an internal node $u'$ with $X_{u'}=X$ in $\cT$ and a Weiner link from $u'$ to $u$. Since $|X_{u'}| = |X_u|-1$, it holds by the
inductive hypothesis that $u'$ belongs to $\cT_w$, and thus $u$ belongs to 
$\cT_w$ as a child of $u'$.
\end{proof}

We do not build $\cT_w$ explicitly, but just traverse its nodes conceptually in
depth-first order and compute the symbols to store in the structures $D_u$ 
and the intervals in $\oB$. 
Let $u$ be the current node of $\cT$ in this traversal and $\ou$ its
corresponding locus in $\ocT$. Assume for now that $\ou$ is a nod, too.
Let $[l_u,r_u]$ be the interval of $X_u$ in $B$ 
and $[l_\ou,r_\ou]$ be the interval of the reverse string $X_\ou$ in $\oB$.%
\footnote{In the rest of the paper we wrote $\oB[l_u..r_u]$ instead of
$\oB[l_\ou..r_\ou]$ for simplicity, but this may cause confusion in this
section.}
Our algorithm starts at the root nodes of $\cT_w$, $\cT$, and $\ocT$, which 
correspond to the empty string, and the intervals in $B$ and $\oB$ are 
$[l_u,r_u]=[l_\ou,r_\ou]=[0,n-1]$. 
We will traverse only
the heavy nodes, yet in some cases we will have to work on all the nodes. We
ensure that on heavy nodes we work at most $O(\log\sigma)$ time, and at most
$O(1)$ time on arbitrary nodes.

Upon arriving at each node $u$, we first compute its heavy children. From the 
topology of $\cT$ we identify the interval $[l_i,r_i]$ for every child
$u_i$ of $u$, by counting leaves in the subtrees of the successive children of
$u$. By reporting all the distinct symbols in $\oB[l_\ou..r_\ou]$ with their
frequencies, we identify the labels of those children. However, the labels are 
retrieved in arbitrary order and we cannot afford sorting them all. Yet, since
the labels are associated with their
frequencies in $\oB[l_\ou..r_\ou]$, which match their number of leaves in the
subtrees of $u$, we can discard the labels of the light children, that is,
those appearing less than $d$ times in $\oB[l_\ou..r_\ou]$. The remaining,
heavy, children are then sorted and associated with the successive heavy
children $u_i$ of $u$ in $\cT$. 

If our preliminary pass marked that a $D_u$ structure must be 
built, we construct at this moment the deterministic dictionary 
\cite{HagerupMP01} with the labels $a$ of the heavy children of $u$ we have just
identified, and associate them with the satellite data $\ra_a(l_\ou-1,\oB)$ 
and $\ra_a(r_\ou,\oB)$. This construction takes $O(\log\sigma)$ time per
element, but it includes only heavy nodes.

We now find all the Weiner links from $u$. For every (heavy or light)
child $u_i$ of $u$, we compute the list $L_i$ of all the distinct symbols that 
occur in $B[l_i..r_i]$.  We mark those symbols $a$ in an array
$V[0..\sigma-1]$ that holds three possible values: not seen, seen, and seen (at
least) twice. If $V[a]$ is not seen, then we mark it as seen; if it is seen, we 
mark it as seen twice; otherwise we leave it as seen
twice. We collect a list $E_u$ of the symbols that are seen twice along this
process, in arbitrary order. For every symbol $a$ in $E_u$, there 
is an explicit Weiner link from $u$ labeled by $a$: Let $X=X_u$; if $a$ 
occurred in $L_i$ and $L_j$ then both $aXa_i$ and $aXa_j$ occur in $T$ and 
there is a suffix tree node that corresponds to the string $aX$. 
The total time to build $E_u$ amortizes to $O(n)$: for each child
$v$ of $u$, we pay $O(1)$ time for each child the node
$\ov$ has in $\ocT$; each node in
$\ocT$ contributes once to the cost.

The targets of the Weiner links from $u$ in $\cT$ correspond to the children of the node $\ou$ in 
$\ocT$. To find them, we collect all the distinct symbols in $B[l_u..r_u]$ and 
their frequencies. Again, we discard the symbols with frequency less than $d$,
as they will lead to light nodes, which we do not have to traverse. The others
are sorted and associated with the successive heavy children of $\ou$.
By counting leaves in the successive children, we obtain the intervals 
$\oB[l'_i..r'_i]$ corresponding to the heavy children $\ou'_i$ of $\ou$.

We are now ready to continue the traversal of $\cT_w$: for each Weiner link
from $u$ by symbol $a$ leading to a heavy node, which turns out to be the $i$-th
child of $\ou$, we know that its node in $\ocT$ is $\ou'_i$ (computed from $\ou$
using the tree topology) and its interval is $\oB[l'_i..r'_i]$. To compute the
corresponding interval on $B$, we use the backward step operation, 
$B[x,y] = B[\Acc[a]+\ra_a(l_u-1,B),\Acc[a]+\ra_a(r_u,B)-1]$. This requires
$O(\log\log_w \sigma)$ time, but applies only to heavy nodes. Finally, the
corresponding node in $\cT$ is obtained in constant time as the lowest common
ancestor of the $x$-th and the $y$-th leaves of $\cT$.

In the description above we assumed for simplicity that $\ou$ is a node in $\ocT$. In the general case $\ou$ can be located on an edge of $\ocT$. This situation arises when all occurrences of $\oXu$ in the reverse text $\oT$ are followed by the same symbol $a$. In this case there is at most one Weiner link from $u$; the interval in $\oB$ does not change as we follow that link.

A recursive traversal of $\cT_w$ might require $O(n \sigma\log n)$ bits for 
the stack, because we store several integers associated to heavy children during
the computation of each node $u$. We can limit the stack height by determining 
the largest subtree among the Weiner links of $u$, traversing all the others 
recursively, and then moving to that largest Weiner link target without 
recursion \cite[Lem.\ 1]{BelazzouguiCKM13}. Since only the largest subtree of 
a Weiner 
link target can contain more than half of the nodes of the subtree of $u$, the 
stack is guaranteed to be of height only $O(\log n)$. The space usage is thus 
$O(\sigma\log^2 n) = O(n\log\sigma)$. 

As promised, we have spent at most $O(\log\sigma)$ time on heavy nodes, which
are $O(n/d)=O(n/\log \sigma)$ in total, thus these costs add up to $O(n)$.
All other costs that apply to arbitrary nodes are $O(1)$. The structures for 
partial rank queries (and the succinct SB-trees) can also be built in linear 
deterministic time, as shown in Section~\ref{sec:smallrank}. 
Therefore our index can be constructed in $O(n)$ time. 

\section{A Compact Index}

As an application of our techniques, we show that it is possible to obtain
$O(|P|/\log_\sigma n + \log^2 n)$ search time, and even
$O(|P|/\log_\sigma n + \log n (\log\log n)^2)$, with an index that
uses $O(n\log\sigma)$ bits and is built in linear deterministic time.

We store $\oB$ in compressed form and a sample of the heavy nodes of $\cT$.
Following previous work \cite{GrossiV05,NN17}, we start from the root and store 
a deterministic dictionary \cite{HagerupMP01} with all the highest suffix tree
nodes $v$ representing strings of depth $\ge \ell = \log_\sigma n$. The key 
associated with each node is a $\log(n)$-bit integer formed with the first 
$\ell$ symbols of the strings $P_v$. The satellite data are the length $|P_v|$,
a position where $P_v$ occurs in $T$, and the range $\oB[l_v..r_v]$ of $v$.
From each of those nodes $v$, we repeat the process with the 
first $\ell$ symbols that follow after $P_v$, and so on.
The difference is that no light node will be inserted in those dictionaries.
Let us charge the $O(\log n)$ bits of space to the children nodes, which are 
all heavy. If we count only the special nodes, which are $O(n/d)$, this
amounts to $O((n \log n)/d)$ total bits and construction time. 
Recall that $d$ is the maximum subtree size of light nodes. This time will 
use $d=\Theta(\log n)$ to have linear construction time and bit space, and 
thus will not take advantage of small rank interval queries.

There are, however, heavy nodes that are not special. These form possibly long
chains between special nodes, and these will also induce chains of sampled
nodes. While building the dictionaries for those
nodes is trivial because they have only one sampled child, the total space
may add up to $O(n\log\sigma)$ bits, if there are $\Theta(n)$ heavy nodes 
and the sampling chooses one out of $\ell$ in the chains. To avoid this,
we increase the sampling step in those chains, enlarging it to $\ell'=\log n$.
This makes the extra space spent in sampling heavy non-special nodes to be
$O(n)$ bits as well.

In addition, we store the text $T$ with a data structure that uses $nH_k(T)
+o(n\log\sigma)$ for any $k=o(\log_\sigma n)$, and allows us extract 
$O(\log_\sigma n)$ consecutive symbols in constant time \cite{FV07}.

The search for $P$ starts at the root, where its first $\ell$ symbols are used 
to find directly the right descendant node $v$ in the dictionary stored at the 
root. If $|P_v| > \ell$, we directly compare the other $|P_v|-\ell$ symbols of 
$P$ with the text, from the stored position where $v$ appears, by chunks of 
$\ell$ symbols. Then we continue the search from $v$ (ignore the chains for
now).

When the next $\ell$ symbols of $P$ are not found in the dictionary of the
current node $v$, or 
there are less than $\ell$ remaining symbols in $P$, we continue using backward
search on $\oB$ from the interval $\oB[l_v..r_v]$ stored at $v$, and do not
use the suffix tree topology anymore. If there are less 
than $\ell$ remaining symbols in $P$, we just proceed symbolwise and complete 
the search in $O(\ell \log\log_w \sigma)$ time. Before this point, we still 
proceed by chunks of $\ell$ symbols, except when (conceptually) traversing 
the explicit light suffix tree nodes, where we perform individual backward
steps. Because there are only $O(d)$ suffix tree nodes in the remaining subtree,
the number of backward steps to traverse such light nodes is $O(d)$ and can be 
performed in time $O(d \log\log_w \sigma)$. All the symbols between consecutive
light nodes must be traversed in chunks of $\ell$. Let $v$ be our current node
and $u$ its desired child, and let $k+1$ be the number of symbols labeling the 
edge between them. After the first backward step from $v$ to $u$ leads us from 
the interval $\oB[l_v..r_v]$ to $\oB[l..r]$, we must perform $k$ further 
backward steps, where each intermediate interval is formed by just one symbol. 
To traverse them fast, we first compute $t_l = \overline{\SA}[l]$ and
$t_r = \overline{\SA}[r]$. The desired symbols are then
$\oT[t_l-1]=\oT[t_r-1], \ldots, \oT[t_l-k]=\oT[t_r-k]$. 
We thus compare the suffixes $T[n-t_l ..]$ 
and $T[n-t_r ..]$ with what remains of $P$, by chunks of $\ell$ symbols, until 
finding the first difference; $k$ is then the number of coincident symbols 
seen. If the two suffixes coincide up to $k+1$ but they differ from $P$, then
$P$ is not in $T$. Otherwise, we can move (conceptually) to node $u$ and 
consume the $k$ coincident characters from $P$. We can compute the new interval
$[l_u..r_u] = [\overline{\SA}\,^{-1}[t_l-k],\overline{\SA}\,^{-1}[t_r-k]]$.
Since $\overline{\SA}$ and its inverse are computed in $O(b)$ time with arrays
similar to $SAM_b$, we spend $O(b)$ time to cross each of the $O(d)$ edges in 
the final part of the search.

Let us now regard the case where we reach a sampled node $v$ that starts a 
sampled chain. In this case the sampling step grows to $\ell'$. We still 
compare $P$ with a suffix of $T$ where its heavy sampled child $u$ appears,
in chunks of $\ell$ symbols. If they coincide, we continue the search from
$u$. Otherwise, we resume the search from $\oB[l_v..r_v]$ using backward
search. We might have to process $\ell'$ symbols of $P$, in time
$O(\ell' \log\log_w \sigma)$, before reaching a light node. Then, we proceed
as explained.

Overall, the total space is $2nH_k(T)+o(n\log\sigma)+O(n) + O((n\log n)/b)$ 
bits, and the construction time is $O(n)$. To have $o(n\log\sigma)+O(n)$ bits of
redundancy in total, we may choose $b=\Theta(\log n)$ or any 
$b=\omega(\log_\sigma n)$. The search time is $O(|P|/\log_\sigma n +
\ell' \log\log_w \sigma + d \log\log_w \sigma + db)$. We can, for
example, choose $b=\Theta(\log n)$, to obtain the following result.

\begin{theorem}
  \label{thm:second}
On a RAM machine of $w=\Omega(\log n)$ bits, we can construct an index for a 
text $T$ of length $n$ over an alphabet of size $\sigma=O(n/\log n)$ in $O(n)$ 
deterministic time using $O(n\log\sigma)$ bits of working space. This index 
occupies $2nH_k(T)+ o(n\log\sigma)+O(n)$ bits of space for any 
$k=o(\log_\sigma n)$. The occurrences of a pattern string $P$ can be counted in 
$O(|P|/\log_\sigma n+\log^2 n)$ time, and then each such occurrence can be
located in $O(\log n)$ time. An arbitrary substring $S$ of $T$  can be 
extracted in time $O(|S|/\log_\sigma n)$.
\end{theorem}

\subsection{Faster and Larger}

By storing the BWT $B$ of $T$ and the $O(n)$ additional bits to support a
compressed suffix tree $\cT$ \cite{Sadakane07}, we can reduce the $O(\log^2 n)$
extra time to $O(\log n (\log\log n)^2)$.

The idea is to speed up the traversal on the light nodes, as follows. For each
light node $v$, we store the leaf in $\cT$ where the heavy path starting at $v$
ends. The heavy path chooses at each node the subtree with the most leaves, thus
any traversal towards a leaf has to switch to another heavy path only 
$O(\log d)$
times. At each light node $v$, we go to the leaf $u$ of its heavy path, obtain
its position in $T$ using the sampled array $SAM_b$ of $B$, and compare the
rest of $P$ with the corresponding part of the suffix, by chunks of $\ell$
symbols. Once we determine the number $k$ of symbols that coincide with $P$
in the path from $v$ to $u$, we perform a binary search for the highest ancestor
$v'$ of $u$ where $|P_{v'}-P_v| \ge k$. If $|P_{v'}-P_v| > k$, then $P$
does not appear in $T$ (unless $P$ ends at the $k$-th character compared, in
which case the locus of $P$ is $v'$). Otherwise, we continue the search from 
$v'$. Each binary search step requires $O(b)$ time to determine $|P_x|$
\cite{Sadakane07}, so it takes time $O(b \log d)$. Since we switch to another 
heavy path (now the one starting at $v'$) $O(\log d)$ times, the total time is 
$O(b \log^2 d) = O(\log n (\log\log n)^2)$ instead of $O(bd) =O(\log^2 n)$.

To store the leaf $u$ corresponding to each light node $v$, we record the
difference between the preorder numbers of $u$ and $v$, which requires 
$O(\log d)$ bits. The node $u$ is easily found in constant time from this
information \cite{Sadakane07}.
We have the problem, however, that we spend $O(\log d)=O(\log\log n)$ bits 
per light node, which adds up to $O(n\log\log n)$ bits. To reduce this to
$O(n)$, we choose a second sampling step $e = O(\log\log n)$, and do not store
this information on nodes with less than $e$ leaves, which are called 
light-light. Those light nodes with $e$ leaves or more are called
light-heavy, and those with at least two light-heavy children are called 
light-special. There are $O(n/e)$ light-special nodes. We store heavy 
path information only for light-special nodes or for light-heavy nodes that
are children of heavy nodes; both are $O(n/e)$ in total. A light-heavy node
$v$ that is not light-special has at most one light-heavy child $u$, and the 
heavy path that passes through $v$ must continue towards $u$. Therefore, if 
it turns out that the search must continue from $v$ after the binary search on 
the heavy path, then the search must continue towards the light-light children
of $v$, therefore no heavy-path information is needed at node $v$.

Once we reach the first light-light node $v$, we proceed as we did for 
Theorem~\ref{thm:second} on light nodes, in total time
$O(eb) = O(\log n \log\log n)$. We need, however, the interval
$\oB[l_\ov,r_\ov]$ before we can start the search from $v$. The interval
$B[l_v,r_v]$ is indeed known by counting leaves in $\cT$, and we can compute
$l_\ov = \overline{\SA}\,^{-1}[n-1-\SA[l_v]]$ and
$r_\ov = \overline{\SA}\,^{-1}[n-1-\SA[r_v]]$, in time $O(b)$.

\begin{theorem}
  \label{thm:third}
On a RAM machine of $w=\Omega(\log n)$ bits, we can construct an index for a
text $T$ of length $n$ over an alphabet of size $\sigma=O(n/\log n)$ in $O(n)$
deterministic time using $O(n\log\sigma)$ bits of working space. This index
occupies $3nH_k(T)+ o(n\log\sigma)+O(n)$ bits of space for any
$k=o(\log_\sigma n)$. The occurrences of a pattern string $P$ can be counted in
$O(|P|/\log_\sigma n+\log n(\log\log n)^2)$ time, and then each such occurrence
can be located in $O(\log n)$ time. An arbitrary substring $S$ of $T$  can be
extracted in time $O(|S|/\log_\sigma n)$.
\end{theorem}

\section{Conclusions}

We have shown how to build, in $O(n)$ deterministic time and
using $O(n\log\sigma)$ bits of working space, a compressed self-index for
a text $T$ of length $n$ over an alphabet of size $\sigma$ that
searches for patterns $P$ in time $O(|P|+\log\log_w \sigma)$, on a $w$-bit
word RAM machine.
This improves upon previous compressed self-indexes requiring
$O(|P|\log\log\sigma)$ \cite{BCGNNalgor13} or 
$O(|P|(1+\log_w \sigma))$ \cite{BelazzouguiN15} time, 
on previous uncompressed indexes requiring
$O(|P|+\log\log\sigma)$ time \cite{FG13} (but that supports dynamism),
and on previous compressed self-indexes requiring $O(|P|(1+\log\log_w\sigma))$
time and randomized construction (which we now showed how to build in linear
deterministic time) \cite{BelazzouguiN15}. The only indexes offering better
search time require randomized construction 
\cite{BelazzouguiN14,GrossiV05,NN17} or $\Theta(n\log n)$ bits of space
\cite{NN17,BGS17}.

As an application, we showed that using $O(n\log\sigma)$ bits of space, we can
build in $O(n)$ deterministic time an index that searches in time
$O(|P|/\log_\sigma n+\log n(\log\log n)^2)$.

It is not clear if $O(|P|)$ time, or even $O(|P|/\log_\sigma n)$, query time
can be achieved with a linear deterministic construction time, even if we
allow $O(n\log n)$ bits of space for the index (this was recently approached,
but some additive polylog factors remain \cite{BGS17}). This is the most 
interesting open problem for future research.

\bibliography{long}

\end{document}